\documentclass[
10pt, 
a4paper, 
oneside, 
headinclude,footinclude, 
BCOR0mm, 
]{scrartcl}

\usepackage[
nochapters, 
pdfspacing, 
dottedtoc 
]{classicthesis} 
\usepackage{arsclassica} 
\usepackage[utf8]{inputenc} 
\usepackage[left=3cm,top=3cm,right=3cm,bottom=3cm,bindingoffset=0.0cm]{geometry}

\usepackage{amsmath,amssymb,amsthm} 

\theoremstyle{definition} 

\theoremstyle{plain} 

\theoremstyle{remark}

\usepackage{graphicx}
\usepackage{listings}
\usepackage{mathtools}
\usepackage{bigints}
\usepackage{xcolor}
\usepackage{enumitem}
\usepackage{eso-pic}
\usepackage{tikz}
\usepackage{eurosym}
\usepackage{easybmat}
\usepackage{multirow,bigdelim}
\usepackage{booktabs}
\usepackage{tikz}
\usepackage{siunitx}
\usepackage{booktabs}
\usepackage{array}
\usepackage{longtable}
\usepackage{float}

\theoremstyle{plain}
\newtheorem*{theorem*}{Theorem}

\hypersetup{
	colorlinks=true, breaklinks=true, bookmarks=true,bookmarksnumbered,
	urlcolor=webbrown, linkcolor=RoyalBlue, citecolor=webgreen, 
	pdftitle={}, 
	pdfauthor={\textcopyright}, 
	pdfsubject={}, 
	pdfkeywords={}, 
	pdfcreator={pdfLaTeX}, 
	pdfproducer={LaTeX with hyperref and ClassicThesis} 
}

\newcolumntype{C}[1]{>{\centering\arraybackslash}p{#1}}

\newcommand\bovermat[2]{%
	\makebox[0pt][l]{$\smash{\overbrace{\phantom{%
					\begin{matrix}#2\end{matrix}}}^{\text{#1}}}$}#2}

\newcommand\scalemath[2]{\scalebox{#1}{\mbox{\ensuremath{\displaystyle #2}}}}

\title{\normalfont\spacedallcaps{Markov Chain-based Cost-Optimal Control Charts
		for Healthcare Data}}
\author{\spacedlowsmallcaps{Bal\'{a}zs Dobi \& Andr\'{a}s Zempl\'{e}ni} \\
	{\small Department of Probability Theory and Statistics, E\"{o}tv\"{o}s Lor\'{a}nd University,
		Budapest, Hungary} \\
	{\small Faculty of Informatics, University of Debrecen, Debrecen, Hungary}}

\date{}

\usepackage[T1]{fontenc}
\usepackage{fouriernc}
\usepackage[scale=0.92]{tgschola}

\usepackage[superscript,biblabel]{cite}

\begin{document}
	
	\bibliographystyle{unsrt}
	
	\renewcommand{\sectionmark}[1]{\markright{\spacedlowsmallcaps{#1}}}
	\lehead{\mbox{\llap{\small\thepage\kern1em\color{halfgray}
				\vline}\color{halfgray}\hspace{0.5em}\rightmark\hfil}}
	
	\pagestyle{scrheadings}
	
	
	\maketitle
	\thispagestyle{empty}
	\setcounter{tocdepth}{2}
	\tableofcontents
	\section*{Abstract}
	Control charts have traditionally been used in industrial statistics, but are
	constantly seeing new areas of application, especially in the age of Industry
	4.0. This paper introduces a new method, which is suitable for applications in
	the healthcare sector, especially for monitoring a health-characteristic of a
	patient. We adapt a Markov chain-based approach and develop a method in which
	not only the shift size (i.e. the degradation of the patient's health) can be
	random, but the effect of the repair (i.e. treatment) and time between samplings
	(i.e. visits) too. This means that we do not use many
	often-present assumptions which are usually not applicable for medical
	treatments. The average cost of the protocol, which is determined by the time
	between samplings and the control limit, can be estimated using the stationary
	distribution of the Markov chain.
	
	Furthermore, we incorporate the standard deviation of the cost into the
	optimisation procedure, which is often very important from a process control
	viewpoint. The sensitivity of the optimal parameters and the resulting average
	cost and cost standard deviation on different parameter values is investigated.
	We demonstrate the usefulness of the approach for real-life data of patients
	treated in Hungary: namely the monitoring of cholesterol level of patients with
	cardiovascular event risk. The results showed that the optimal parameters from
	our approach can be somewhat different from the original medical parameters.
	\vskip 0.5cm
	\noindent KEYWORDS: control chart; cost-effectiveness; Markov-chain;
	healthcare

	\clearpage
	\pagenumbering{arabic} 
	\section{Introduction}
	
	Statistical process control, and with it control charts enjoy a wide range of
	use today, and have seen great developments, extensions and generalisations
	since the original design of Shewhart.\cite{Montgomery} This proliferation of
	methods and uses can at part be attributed to the fact that information is
	becoming available in ever increasing quantities and in more and more areas.
	
	Even though control charts have originally been designed with statistical
	criteria in mind, the development of cost-efficient control charts also began
	early. Duncan in his influential work from 1956 developed the basis for a cycle
	based cost-effective control chart framework which is still very popular today
	and is implemented in statistical software packages such as in
	\textbf{\textsf{R}}.\cite{Duncan,Mortarino,Zhu} Cost-efficient solutions are
	often the focus in industrial and other settings besides statistical optimality.
	
	Statistical process control, including control charts can be found in very
	different environments in recent literature. For example in mechanical
	engineering, e.g. Zhou et al.\cite{Zhou}, where the authors develop
	a $T^2$ bootstrap control chart, based on recurrence plots. Another example is
	the work of Sales et al.\cite{Sales} at the field of chemical engineering, in
	their work they use multivariate tools for monitoring soybean oil
	transesterification.  It is not a surprise that the healthcare sector has seen
	an increased use of control chart too. Uses within this sector range from
	quality assurance to administrative data analysis and patient-level monitoring,
	and many more.\cite{Thor,Suman} In one example control charts were used as part
	of system engineering methods applied to healthcare delivery.\cite{Padula}
	Other examples include quality monitoring in thyroid surgery\cite{Duclos},
	monitoring quality of hospital care with administrative data\cite{Coory} and
	chronic respiratory patient control\cite{Correia}.
	
	Cost-efficient control charts have not been widely used in healtcare settings,
	but there are some examples which deal with cost monitoring and
	management. In one work $\overline{X}$ and $R$ charts were used to assess the
	effect of a physician educational program  on the hospital resource
	consumption\cite{Johnson}. Another article is about a hospital which used
	control charts for monitoring day shift, night shift and daily total staffing
	expense variance, among other variables\cite{Shaha}. Yet another
	paper documents a case study about primary care practice performance, where
	control charts were used to monitor costs associated with provider productivity
	and practice efficiency. Further costs monitored were net patient revenue per
	relative value unit, and provider and non‐provider cost as a percent of net
	revenue\cite{Stewart}. Even though these studies used control charts for
	cost monitoring or optimisation purposes, they did not deal with the same
	problems as this paper, as our method focuses on cost-optimisation by finding
	the optimal parameters of the control chart setup.
	
	The aim of this study is to present a cost-efficient control chart framework
	which was specifically designed with use on healthcare data in mind.
	Specifically, we would like to use control charts for the purposes of analysing
	and controlling a healthcare characteristic of a patient over time, such as the
	blood glucose level. Traditionally, the minimal monitoring and process cost is
	achieved by finding the optimal parameters, namely the sample size, time between
	samplings and critical values\cite{Montgomery}. If one desires to
	find the optimal parameters for a cost-optimal control chart for a healthcare
	process, then the proper modelling of the process is critical, since  if the
	model deviates greatly from reality, then the resulting parameters and costs may
	not be appropriate.
	Of course this presents several problems as certain typical assumptions
	in control chart theory will not hold for these kind of processes.
	Namely the assumption of fixed and known shift size, which is problematic
	because healthcare processes can often be drifting in nature and can produce
	different shift sizes. Another assumption is the perfect repair, which in this
	case means that the patient can always be fully healed, which is obviously
	often impossible.
	Lastly, the time between shifts is usually set to be constant, but in a
	healthcare setting with human error and non-compliance involved, this also
	needs to be modified. Furthermore, since these processes can have undesirable
	erratic behaviour, the cost standard deviation also needs to be taken into
	account besides the expected cost.
	
	The control chart design which is presented here for dealing with these
	problems is the Markov chain-based framework. This framework was developed by
	Zempléni et al. in 2004 and was successfully applied on
	data collected from a Portuguese pulp plant\cite{Zempleni}. Their article
	provided suitable basis for generalisations.
	Cost-efficient Markov chain-based control charts build upon Duncan's cycle
	based model, as these also partition the time into different segments for cost
	calculation purposes. The previous paper from 2004\cite{Zempleni} has already
	introduced the random shift size element, in this paper we expand upon that idea and develop
	further generalisations to create cost-efficient Markov chain-based control
	charts for healthcare data.
	
	The article is organized in the following way: Section 2 starts with
	basic definitions and notions needed to understand the more complicated
	generalisations. Subsection 2.2 discusses the mathematics behind
	the aforementioned new approaches in the Markov chain-based framework.
	Subsection 2.3 deals with the problems introduced by discretisation - which is needed for the
	construction of the transition matrix - and also explains the calculation of
	the stationary distribution, and defines the cost function.
	Analysis of results and examples will be provided with the help of program
	implementations written in \textbf{\textsf{R}}. Section 3 shows an example of
	use for the new control chart methods in healthcare settings. This example
	involves the monitoring of low-density lipoprotein levels of patients at risk of
	cardiovascular events. Section 4 concludes the theoretical and applied results.
	
	
	\section{Methods}
	
	The methods and \textbf{\textsf{R}} implementation in this paper largely depend
	on the works of Zempléni et al. The first part of this section gives
	a brief introduction to the methods used in their paper.\cite{Zempleni}
	\subsection{The Markov-chain-based Framework}
	Consider a process which is monitored by a control chart. We will
	show methods for processes which can only shift to one (positive)
	direction, monitored by a simple $X$-chart, with sample size $n=1$. This aims to
	model the monitoring of a characteristic where the shift to only one direction
	is interesting or possible, and the monitoring is focused at one patient at a
	time.
	Several assumptions are made for the base model.
	The process distribution is normal with known parameters $\mu_0$ and
	$\sigma$. We will denote its cumulative distribution function by $\phi$.
	The shift intensity $1/s$ is constant and known and the shift size $\delta^*$ is
	fixed. It is also assumed that the process does not repair itself, but when
	repair is carried out by intervention, it is perfect. The repair is treated as
	an instantaneous event. All costs related to repairing should be included in the repair cost,
	for example if a major repair entails higher cost, then this should also be
	reflected in the calculation. The time
	between shifts is assumed to be exponentially distributed. The above
	assumptions ensure that the probabilities of future transitions are only
	dependent on the current state. This is the so-called Markov property, and the
	model under consideration is called a Markov chain.
	The states of this Markov chain are defined at the sampling times and the type
	of the state depends on the measured value and the actual (unobservable)
	background process, namely whether there was a shift from the target value in
	the parameter. This way four basic state types are defined:
	\begin{itemize}
		\item No shift - no alarm: in-control (INC)
		\item Shift - no alarm: out-of-control (OOC)
		\item No shift - alarm: false alarm (FA)
		\item Shift - alarm: true alarm (TA)
	\end{itemize}
	A graphical representation of this process can be seen on Figure 1.
	\begin{figure}[ht]
		\caption{Definition of States}
		\includegraphics[scale=0.75]{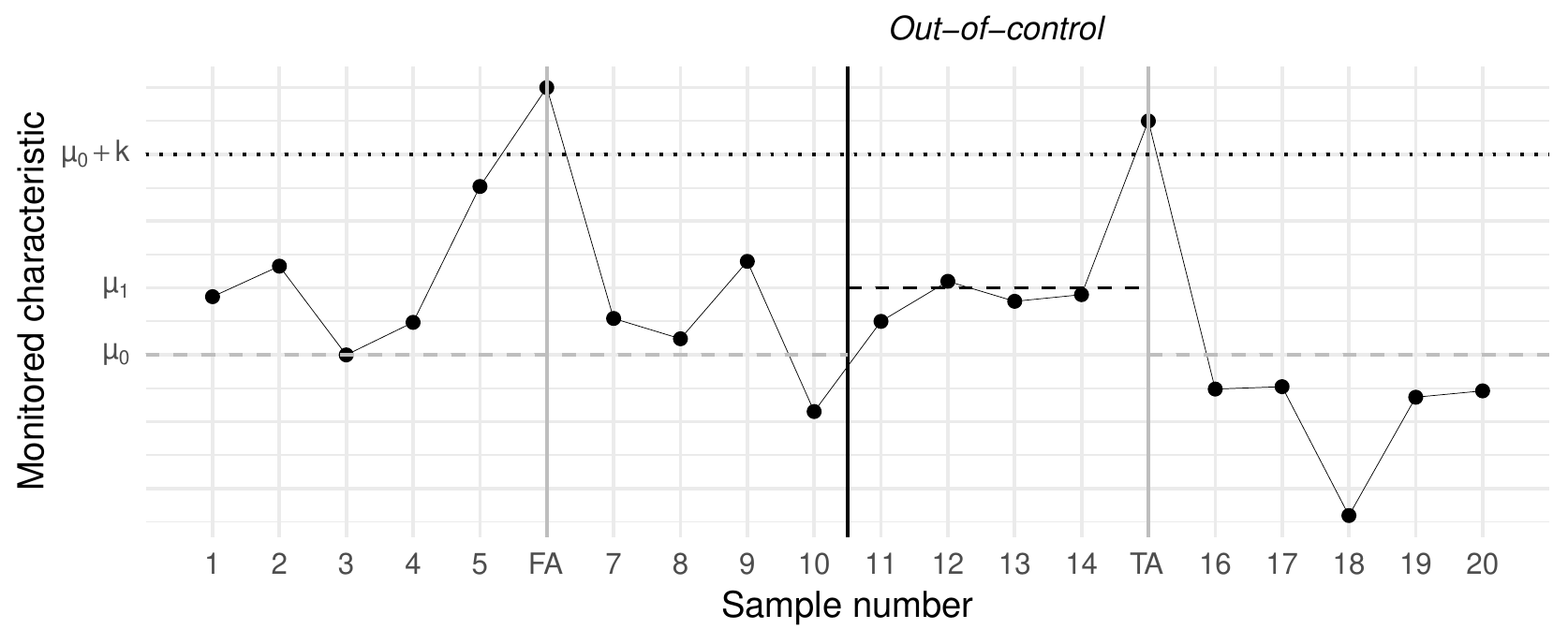}
		\centering
	\end{figure}
	The transition matrix of this Markov chain can be seen below.
	\vspace{5mm}
	\begin{center}
		$\begin{matrix}
		\scalemath{1}{
			\begin{bmatrix}
			\bovermat{In-control}{ (1-F(h))\phi(k)}  &
			\bovermat{Out-of-control}{F(h))\phi(k-\delta^*)} & \bovermat{False
				alarm}{(1-F(h))(1-\phi(k))}  & \bovermat{True
				alarm}{F(h)(1-\phi(k-\delta^*))} \\
			0 & \phi(k-\delta^*) & 0 & 1-\phi(k-\delta^*) \\
			(1-F(h))\phi(k) & F(h))\phi(k-\delta^*) & (1-F(h))(1-\phi(k)) &
			F(h)(1-\phi(k-\delta^*)) \\
			(1-F(h))\phi(k) & F(h))\phi(k-\delta^*) & (1-F(h))(1-\phi(k)) &
			F(h)(1-\phi(k-\delta^*))
			\end{bmatrix}}
		\end{matrix}$
	\end{center}
	Here $F()$ is the cumulative distribution function of the exponential
	distribution with expectation $1/s$, where $s$ is the expected number of shifts
	in a unit time interval. $h$ is the time between consecutive observations and
	$\phi$ is the cumulative distribution function of the process.
	$k$ is the control limit and $\delta^*$ is the size of the shift.
	After calculating the stationary distribution of this Markov chain, one can
	define a cost function for the average total cost:
	\begin{equation}
	E(C) =\frac{c_s + p_3c_f + p_4c_r}{h} + p_2c_o + p_4c_oB
	\end{equation}
	where $p_i$ is the probability of state $i$, in the stationary distribution.
	$c_s$, $c_f$, $c_o$ and $c_r$ are the cost of sampling, false alarm,
	out-of-control operation and repair respectively. $B$ is the fraction of time
	between consecutive observations where the shift occurred and remained
	undetected\cite{Mortarino}:
	\begin{equation}
	B=\frac{hse^{hs} - e^{hs} + 1}{hs(e^{hs} - 1)} \nonumber
	\end{equation}
	An economically optimal chart, based on the time between samplings and the
	control limit can be constructed by minimising the cost function.
	
	\subsection{Generalisation}
	The previous simple framework can be used for more
	general designs. Zempléni et al.\cite{Zempleni} used this method to set up
	economically optimal control charts where shifts are assumed to have a random
	size and only the distribution of the shift size is known. This
	means that the shift size is no longer assumed to be fixed and known, which is important in
	modeling various processes not just in healthcare, but in industrial or
	engineering settings too. This requires expanding the two shifted states to
	account for different shift sizes.
	Let $\tau_i$ denote the random shift times on the
	real line and let $\rho_i$ be the shift size at time $\tau_i$. Let the
	probability mass function of the number of shifts after time $t$ from the start
	be denoted by $\nu_t$. $\nu_t$ is a discrete distribution with support over
	$\mathbb{N}^0$. Assume that $\rho_i$ follows a continuous distribution, which
	has a cumulative distribution function with support over $(0, \infty)$, and that
	the shifts are independent from each other and from $\tau_i$.
	If the previous conditions are met, the resulting random process has monotone
	increasing trajectories between samplings, which are step functions. The
	cumulative distribution function of the process values for a given time $t$
	from the start can be written the following way:
	\begin{equation}
	Z_t(x)=
	\begin{dcases}
	0 &{\text{if }} x < 0, \\
	\nu_t(0) + \sum_{k=1}^\infty \nu_t(k) \Psi_k(x) &{\text{if }} x \geq 0
	\end{dcases}
	\end{equation}
	where $\Psi_k()$ is the cumulative distribution function of the
	sum of $k$ independent, identically distributed $\rho_i$ shift sizes.
	The $x=0$ case means there is no shift. The probability of zero shift size is
	just the probability of no shift occurring, which is $\nu_t(0)$.
	
	Let us assume now that the shift times form a homogeneous Poisson process, and
	the shift size is exponentially distributed, independently of previous
	events. The choice of the exponential distribution was motivated by the 
	tractability of its convolution powers as a gamma distribution. The cumulative
	distribution function of the shift size by (2) is then:
	\begin{equation}
	Q_t(x)=
	\begin{dcases}
	0 &{\text{if }} x < 0, \\
	n_t(0) + \sum_{k=1}^\infty n_t(k) Y_k(x) &{\text{if
	}} x \geq 0
	\end{dcases}
	\end{equation}
	where $n_t$ is the probability mass function of the Poisson distribution, with
	parameter $ts$ - the expected number of shifts per unit time
	multiplied by the time elapsed - representing the number of shifts in the time
	interval $(0;t)$.
	$Y_k$ - the shift size cumulative distribution function, for $k$ shift events -
	is a special case of the gamma distribution, the Erlang distribution
	$E(k,\frac{1}{\delta})$, which is just the sum of $k$ independent exponential
	variates each with mean $\delta$.
	
	The framework can be generalised even further by not assuming perfect repair
	after a true alarm signal. This means that the treatment will not have
	perfect results on the health of the patient, or - in industrial or engineering
	settings - that the machines cannot be fully repaired to their original
	condition.  In this case, the imperfectly repaired states act as out-of-control
	states. It is assumed that the repair cannot worsen the state of the process,
	but an imperfectly repaired process will still cost the same as an equally
	shifted out-of-control process, thus repaired and out-of-control states do not
	need distinction during the cost calculation. Different approaches can be
	considered for modeling the repair size distribution. The one applied here uses
	a random variable to determine the distance from the target value after repair.
	A natural choice for the distribution of this random variable could be the beta
	distribution since it has support over $[0,1]$ - the portion of the distance
	compared to the current one after repair. Also, the
	flexible shape of its density function can model many different repair
	processes. Because of these considerations we will assume that the proportion of
	the remaining distance from $\mu_0$ after repair - $R$ - is a
	$Beta(\alpha,\beta)$ random variable, with known parameters.
	
	Yet another generalisation is the random sampling time. In certain environments,
	the occurrence of the sampling at the prescribed time is not always guaranteed.
	For example in healthcare, the patient or employee compliance can have
	significant effect on the monitoring, thus it is important to take this into
	account during the modeling too. Here it is modeled in a way, that the sampling
	is not guaranteed to take place - e.g. the patient may not show up for control
	visit. This means that the sampling can only occur according to the sampling
	intervals, for example at every $n$th days, but is not a guaranteed event. One
	can use different approaches when modeling the sampling probability, here we
	consider two cases. The first one assumes that too frequent samplings will
	decrease compliance. This assumption is intended to simulate the situation in which too
	frequent samplings may cause increased difficulty for the staff or the patient
	- leading to decreased compliance.
	The probability of a successful sampling as a function of the prescribed time
	between samplings is modeled using a logistic function:
	\begin{align}
	T^*_h=\frac{1}{1+e^{-q(h-z)}} \nonumber
	\end{align}
	where $q>0$ is the steepness of the curve, $z \in \mathbb{R}$ is the value of
	the sigmoid's midpoint and $h$ is the time between samplings.
	
	In the other approach it is assumed that too
	frequent samplings will decrease compliance and increased distance from
	the target value will increase compliance. This assumption means that a heavily
	deteriorated process or health state will ensure a higher compliance.
	The probability of a successful sampling as a function of the prescribed time
	between samplings and the distance from the target value is modeled using
	a modified beta distribution function:
	\begin{align}
	T^*_h(w) &= P\Bigg(W_h<\frac{w+\zeta^*}{V+2\zeta^*}\Bigg) \nonumber
	\end{align}
	where $W_h$ is a beta distributed $Beta($a$/h,$b$)$ random variable, $w$ is the
	distance from the target value, $h$ is the time between samplings, $V$ is the
	maximum distance from $\mu_0$ taken into account - the distance where we expect
	maximal possible compliance.
	The shifts in the values of $w$ and $V$ are needed to acquire probabilities
	strictly between $0$ and $1$, since deterministic
	behaviour is undesirable even in extreme cases. These shifts are parametrised by
	the $\zeta^*>0$ value, which should typically be a small number.
	It is important to note, that these are just two ways of modeling the sampling
	probability. Other approaches and distributions may be more appropriate
	depending on the process of interest.
	
	The shift size distribution, the repair size distribution and the sampling
	probability, together with the control chart let us model and monitor the
	behaviour of the process.
	The resulting process is monotone increasing between samplings and has a
	downward "jump" at alarm states - as the repair is assumed to be
	instantaneous. Usually a wide range of different cost types are associated with
	the processes and their monitoring, these include the costs associated with
	operation outside the target value. Since the operator of the control chart
	only receives information about the process at the time of the sampling, the
	proper estimation of the process behaviour between samplings is critical.
	Previously, at the perfect repair and non-stackable, fixed shift size model,
	this task was reduced to estimating the time of shift in case of a true alarm or
	out-of-control state, since otherwise the process stayed at the previous value
	between samplings. The estimation of the process behavior with random shift
	sizes and random repair is more difficult.
	
	The expected out-of-control operation cost can be written as the expectation of
	a function of the distance from the target value. At (2) the shift size
	distribution was defined for a given time $t$, but this time we are
	interested in calculating the expected cost for a whole interval. We propose
	the following calculation method for the above problem:
	
	\newtheorem*{prop}{Proposition}
	
	\begin{prop}
		Let $H_j$ be a random process whose trajectories are monotone increasing step
		functions defined by $\tau_i$ shift times and $\rho_i$ shift sizes as in (2),
		with starting value $j \geq 0$.
		
		The expected value of a function of $H_j$ over an $\epsilon$ long interval
		can be written as:
		\begin{equation}
		E_\epsilon(f(H_j)) = \frac{\bigintsss_{t_0}^{t_0+\epsilon} \bigintsss_{0}^{\infty}
			1-Z_t(f^{-1}(x-j)) dx dt}{\epsilon}
		\end{equation}
		where $Z_t()$ is the shift size cumulative distribution function given at (2),
		$f()$ is an invertible, monotonic increasing function over the real line and
		$t_0$ is the start point of the interval.
	\end{prop}
	\begin{proof}
		Let us observe that the inner integral in the numerator - $\int_{0}^{\infty}
		1-Z_t(f^{-1}(x-j)) dx$ - is just the expected value of a function of the shift
		size at time $t$, since we know that if $X$ is a non-negative random variable,
		then $E(X)=\int_0^\infty (1-F(x))dx$, where $F()$ is the cumulative distribution
		function of $X$.
		In other words, the inner integral in the numerator is the expected value the
		process $f(H_j)$, $t$ time after the start. Furthermore, observe that this
		expected value is a continuous function of $t$. We are looking for the
		expectation of $\int_{0}^{\infty} 1-Z_t(f^{-1}(x-j)) dx$ over
		$[t_0,t_0+\epsilon]$, which is (4).
	\end{proof}
	
	For practical purposes we can apply the previous, general proposition to our
	model of Poisson-gamma mixture shift size distribution, thus we assume that the
	shift size distribution is of the form of (3). The connection between the
	distance from the target value and the resulting cost is often assumed not to be
	linear: often a Taguchi-type loss function is used - the loss is assumed to be
	proportional to the squared distance, see e.g. the book of Deming\cite{Deming}.
	Applying this to the above proposition means $f(x)=x^2$. Since we are
	interested in the behaviour of the process between samplings, $t_0=0$ and
	$\epsilon=h$, thus:
	\begin{align}
	E_h(H_j^2) &= \frac{\bigintsss_{0}^{h}  \Big[e^{-ts}j^2 + 
		\Big(\sum_{k=1}^\infty \frac{(ts)^k
			e^{-ts}}{k!} \cdot \bigintsss_{0}^{\infty}  (x+j)^2 \frac
		{(1/\delta)^{k}x^{k-1}e^{-x/\delta}}{(k-1)!} dx \Big) \Big] dt}{h} \nonumber \\
	&= \frac{\bigintsss_{0}^{h}  e^{-ts}j^2 + \sum_{k=1}^\infty \frac{(ts)^k
			e^{-ts}}{k!}\big(k\delta^2+(k\delta + j)^2\big) dt}{h} =
	\frac{\bigintsss_{0}^{h} 2 \delta^2 ts + (\delta ts + j)^2dt}{h} \nonumber \\
	&= hs\delta \Bigg(\delta + \frac{hs\delta}{3} + j\Bigg) + j^2
	\end{align}
	where first we have used the law of total expectation - the condition being
	the number of shifts within the interval. If there is no shift, then the
	distance is not increased between samplings, this case is included by the
	$e^{-ts}j^2$ term before the inner integral. Note that the inner integral is
	just $E(X+j)^2$ for a gamma - namely an Erlang($k,\frac{1}{\delta}$) -
	distributed random variable. When calculating the sum, we used the known
	formulas for $E(N^2)$, $E(N)$ and the Poisson distribution itself - where $N$
	is a Poisson($ts$) distributed random variable.

	\subsection{Implementation}
	\paragraph{Discretisation} For cost calculation
	purposes we would like to find a discrete stationary distribution which approximates the distribution of the monitored characteristic
	at the time of samplings. This requires the discretisation of the above defined
	functions, which in turn will allow us to construct a discrete time Markov
	chain with discrete state space.
	
	A vector of probabilities is needed to represent the shift size
	probability mass function $q_t()$ during a sampling interval:
	\begin{equation}
	q_t(i)=
	\begin{dcases}
	n_t(0) &{\text{if }}i=0, \\
	\sum_{k=1}^\infty n_t(k) \big(Y_k(i\Delta) - Y_k((i-1)\Delta)\big)
	&{\text{if }} i \in \mathbb{N}^{+} \nonumber
	\end{dcases}
	\end{equation}
	where $\Delta$ stands for the length of an interval, one unit after
	discretisation. For $i=0$ the function is just the probability of no shift
	occurring.
	
	The discretised version of the repair size distribution can be written the
	following way:
	\begin{equation}
	R(l,m) = P\Bigg(\frac{m}{l+1/2}\le R <\frac{m+1}{l+1/2}\Bigg) \nonumber
	\end{equation}
	where $l$ is the number of discretised distances closer to $\mu_0$ than the
	current one - including $\mu_0$. $m$ is the index for the repair size
	segment we are interested in, with $m=0$ meaning perfect repair ($m\leq l$).
	The repair is assumed to move the expected value towards the target value by a
	random percentage, governed by $R()$. Even though discretisation is
	required for practical use of the framework, in reality the repair size
	distribution is continuous. To reflect this continuity in the background, the
	probability of perfect repair is set to be $0$ when the repair is random. $l$
	is set to be $0$ when there is no repair, meaning $R(0,m)\equiv1$.
	The $1/2$ terms are necessary for correcting the overestimation of the
	distances from the target value, introduced by the discretisation: in reality,
	the distance can fall anywhere within the discretised interval, without
	correction the maximum of the possible values would be taken
	into account, which is an overestimation of the actual shift size. After
	correction the midpoint of the interval is used, which can still be somewhat
	biased, but in practical use with fine disretisation this effect is negligible.
	
	When the sampling probability depends on the time between samplings only, the
	model is unchanged, since both the time between samplings and the sampling
	probability can be continuous. However, when the probability also depends
	on the shift size, discretisation is required here as well:
	\begin{equation}
	T_h(v) = P\Bigg(W_h<\frac{v+\zeta}{V_d+\zeta}-\frac{1}{2(V_d+\zeta)}\Bigg)
	\nonumber
	\end{equation}
	now $v$ is the state distance from the target value in discretised units, and
	$V_d$ is the number of considered intervals - discretised shift sizes. The
	$\frac{1}{2(V_d+\zeta)}$ term is necessary for correcting the overestimation of
	the distances from the target value. The denominators contain simply $V_d+\zeta$
	instead of $V_d+2\zeta$, because $v+\zeta$ is already strictly
	smaller than $V_d+\zeta$, since the smallest discretised state is 0 and thus the
	greatest is $V_d-1$. This ensures that the probability can never reach 1.
	Example curves for the successful sampling probability can be seen in Figure 2.
	It shows that longer time between samplings and greater distances from the
	target value increase the probability of successful sampling.
	\begin{figure}[ht]
		\caption{Sampling probabilities for  $q=8$, $z=0.5$
			on the left, and for $\alpha=1$, $\beta=3$, $V_d=100$,
			$\zeta=1$ on the right}
		\includegraphics[scale=0.75]{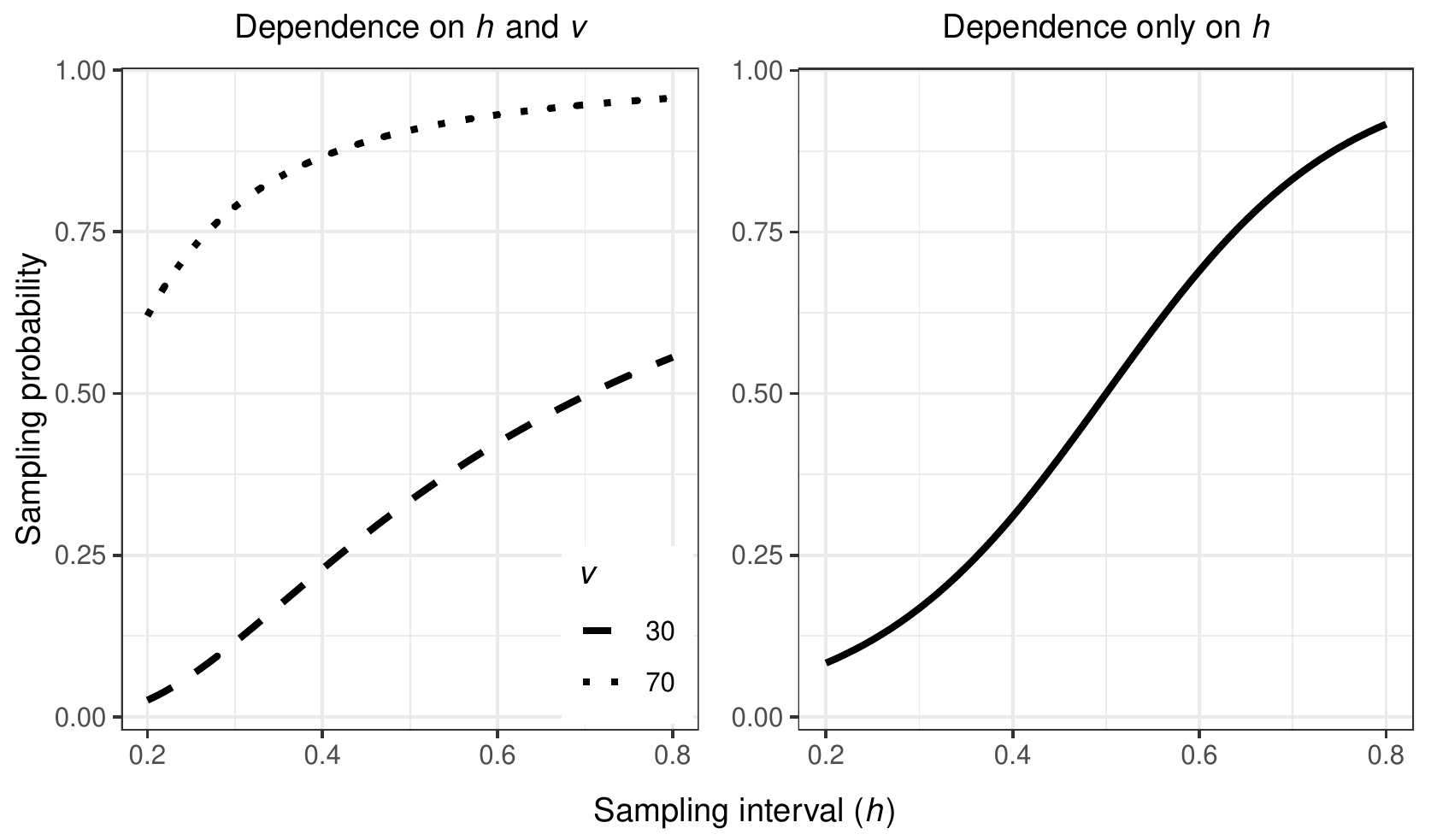}
		\centering
	\end{figure}
	
	\paragraph{Transition Matrix and Stationary Distribution} The transition
	probabilities can be written using the $\phi()$ process distribution, the
	$q_t()$ shift size distribution, the $R()$ repair size distribution and the
	$T_h()$ sampling probability. For the ease of notation let us define the $S()$
	and $S'()$ functions, the first for target states without alarm and the latter
	for target states with alarm:
	\begin{align}
	S(g,v,l)& = \big[T_h(v) \phi(k-\Delta'(v)) + (1-T_h(v))\big]
	\smashoperator[lr]{\sum_{m\in\mathbb{N}^{0}, m\leq l, m\leq v}^{}}
	q_h(g-m)R(l,m) \nonumber \\
	S'(g,v,l)& = T_h(v) \big[1-\phi\big(k-\Delta'(v)\big)\big]
	\smashoperator[lr]{\sum_{m\in\mathbb{N}^{0}, m\leq l, m\leq v}^{}}
	q_h(g-m)R(l,m) \nonumber
	\end{align}
	where $\Delta'(v)$ is a function defined as
	\begin{equation}
	\Delta'(v)=
	\begin{dcases}
	0 &{\text{if }}v=0, \\
	i\Delta-\frac{\Delta}{2} &{\text{if }} v = 1,\hdots,V_d-1 \nonumber
	\end{dcases}
	\end{equation}
	$k-v\Delta$ would simply be the critical value
	minus the size of the shift in consideration. The $-\frac{\Delta}{2}$ term is
	added, because without it, the shift size would be overestimated. The total
	number of discretised disctances is $V_d$, thus the number of non-zero distances
	is $V_d-1$. $g$ is the total shift size possible in discretised units when moving
	to a given state, $v$ is the distance, measured as the number of discretised
	units from the target value, $l$ is the index of the actual partition and $m$
	is the index for the repair size segment we are interested in with $m=0$
	meaning perfect repair. Note, that as before, $l$ is set to be $0$ when
	there is no repair, meaning $R(0,m)\equiv1$.
	$S()$ and $S'()$ are functions which combine the three distributions into the
	probability we are interested in. The possible values of $m$ are restricted
	this way, because the parameter of the $q_h()$ function must be non-negative.
	A more intuitive explanation is, that we assumed positive shifts and a negative
	$g-m$ would imply a negative shift.
	It can be seen that the probability of a state without an alarm - the $S'()$
	function - is increased with the probability of unsuccessful sampling -
	the $1- T_h(v)$ term. Of course it is still possible to not receive an alarm
	even though the sampling successfully occurred.
	
	Using the $S()$ and $S'()$ functions one can construct the transition
	matrix:
	\begin{center}
		\begin{equation}
		\Pi=
		\begin{matrix}
		\scalemath{0.9}{
			\begin{bmatrix}
			\bovermat{In-control}{S(0,0,0)}  & \bovermat{Out-of-control}{S(1,1,0) & S(2,2,0) &
				\dots} & \bovermat{False alarm}{S'(0,0,0)}  &
			\bovermat{True alarm}{S'(1,1,0) &
				S'(2,2,0) & \dots}\\
			0  & S(0,1,0) & S(1,2,0)
			& \dots & 0 & S'(0,1,0) &
			S'(1,2,0) & \dots
			\\
			0 & 0 & S(0,2,0) &
			\dots & 0 & 0 & S'(0,2,0) & \dots
			\\
			\vdots & \vdots & \vdots & & \vdots & \vdots & \vdots & \\
			S(0,0,0) & S(1,1,0) & S(2,2,0) & \dots & S'(0,0,0) & S'(1,1,0) &
			S'(2,2,0) & \dots 
			\\
			0 & S(1,1,1) & S(2,2,1) & \dots & 0 & S'(1,1,1) &
			S'(2,2,1) & \dots
			\\
			0 & S(1,1,2) & S(2,2,2) & \dots & 0 & S'(1,1,2) &
			S'(2,2,2)  & \dots
			\\
			\vdots & \vdots & \vdots & & \vdots & \vdots & \vdots & \ddots
			\end{bmatrix}
		}
		\scalemath{0.66}{
			\begin{aligned}
			&\left.\begin{matrix}
			\\[0.1em]
			\end{matrix} \right\} %
			\\
			&\left.\begin{matrix}
			\\
			\\[1.8em]
			\\
			\end{matrix}\right\}%
			\\
			&\left.\begin{matrix}
			\\[0.1em]
			\end{matrix}\right\}%
			\\
			&\left.\begin{matrix}
			\\
			\\[1.7em]
			\\
			\end{matrix}\right\}%
			\\
			\end{aligned} \nonumber
		}
		\end{matrix}
		\end{equation}
	\end{center}
	The size of the matrix is $2V_d\times2V_d$ since every shift size has two
	states:
	one with and one without alarm. The first $V_d$ columns are states without
	alarm, the second $V_d$ are states with alarm. One can observe, that once the
	process leaves the healthy state it will never return. This is due to the nature of the
	imperfect repair we have discussed.
	
	The transition matrix above defines a Markov chain with a discrete, finite state
	space with one transient, inessential  class (in-control and false alarm states)
	and one positive recurrent class (out-of-control and true alarm states). The
	starting distribution is assumed to be a deterministic distribution
	concentrated on the in-control state, which is to say that the process is
	assumed to always start from the target value. In finite Markov chains, the
	process leaves such a starting transient class with probability one. The
	problem of finding the stationary distribution of the Markov chain is thus
	reduced to finding a stationary distribution within the recurrent classes of
	the chain.
	Since there is a single positive recurrent class which is also aperiodic, we
	can apply the Perron–Frobenius theorem to find the stationary
	distribution\cite{Meyer}:
	\begin{theorem*}
		Let $A$ be an $n \times n$, irreducible matrix with non-negative
		elements, $x\in \mathbb{R}^n$ and $$\lambda_0 = \lambda_0(A) = sup \{ \lambda :
		\exists x > 0 :
		Ax \geq \lambda x \}.$$ Then the following statements hold:
		\begin{itemize}
			\item [1)] $\lambda_0$ is an eigenvalue of $A$ with algebraic multiplicity of
			one, and its corresponding eigenvector $x_0$ has strictly positive elements.
			\item [2)] The absolute value of all other eigenvalues is less than or equals
			$\lambda_0$.
			\item [3)] If $A$ is also aperiodic, then the absolute value of all other
			eigenvalues is less than $\lambda_0$.
		\end{itemize}
	\end{theorem*}
	One can apply this theorem to find the stationary distribution of $\Pi$. If we
	consider now $\Pi$ without the inessential class - let us denote it with $\Pi'$
	- then $\lambda_0(\Pi') = \lambda_0(\Pi'^T) = 1$.
	Moreover, the stationary distribution - which is the left eigenvector of $\Pi'$,
	normalised to sum to one - is unique and exists with strictly positive elements.
	Finding the stationary distribution is then reduced to solving the following
	equation:
	$\Pi'^T f_0 = f_0$,
	where $f_0$ is the left eigenvector of $\Pi'$. This amounts to solving $2V_d-2$
	equations - the number of states minus the in-control and false alarm states -
	for the same number of variables, so the task is easily accomplishable.
	The stationary distribution is then:
	\begin{equation}
	P = \frac{f_0}{\sum _{i{\mathop {=}}1}^{2V_d-2}f_{0_i}}.  \nonumber
	\end{equation}
	
	\paragraph{Cost Function} Using the stationary
	distribution, the expected cost can be calculated:
	\begin{equation}
	E(C) = c_s\frac{1}{h}(T' \cdot P) + \frac{\sum_{i=1}^{V_d-1} \big(c_{rb} +
		c_{rs}\Delta'^2(i)\big) P_{r_i}}{h} + c_o (A^2 \cdot P) \nonumber
	\end{equation}
	This cost function incorporates similar terms as previously $(1)$. The first
	term deals with the sampling cost: $T' = \{ T_h(1), T_h(2),\ldots,T_h(V_d-1),
	T_h(1), T_h(2),\ldots,T_h(V_d-1)\}$ is the vector of successful sampling
	probabilities repeated in a way to fit the length and order of the stationary distribution.
	This first term uses the expected time between samplings, $\frac{1}{h}(T' \cdot
	P)$, instead of simply $h$, which would just be the minimal possible time
	between samplings.
	The second term deals with the repair costs and true alarm probabilities. $P_{r_i}$ is the true alarm
	probability for shift size $i$. The repair cost is partitioned into a base and
	shift-proportionate part:
	$c_{rb}$ and $c_{rs}$. The true alarm probability is used, since it is assumed
	that repair occurs only if there is an alarm. The last term is the average cost
	due to operation while the process is shifted. The connection between the
	distance from the target value and the resulting cost is assumed to be
	proportional to the squared distance.
	This is modeled using the $A^2$ vector, which contains the weighted averages of
	the expected squared distances from the target value between samplings:
	\begin{equation}
	A^2_i = \sum_{j=1}^{i} E_h\big(H_{\Delta'(j)}^2\big) M_{ij} \nonumber
	\end{equation}
	where $E_h\big(H_{\Delta'(j)}^2\big)$ is calculated using (5).  $j$ indicates
	one of the possible starting distances immediately after the sampling, and $i$
	indicates the state - shift - of the process at the current sampling.
	$M_{ij}$ is the probability that $\Delta'(j)$ will be the starting distance
	after the sampling, given that the current state is $i$. These probabilities can
	be written in a matrix form:
	\begin{center}
		\begin{equation}
		\begin{matrix}
		\scalemath{0.9}{
			\begin{bmatrix}
			\bovermat{Distance from the target value starting from 0}{\ \ \ \ \ 0 \
				\ \ \ \ \ & \ \ \ \ 1 \ \ \ \ \ & \ \ \ \ 0 \ \ \ \ \ & \ \ \ \ \ 0 \ \ \
				\ \ \ & \ \ \dots \ \ \ }
			\\
			0 & 0 & 1 & 0 &
			\dots
			\\
			0 & 0 & 0 & 1 & \dots 
			\\
			\vdots & \vdots & \vdots & \vdots & 
			\\
			0 & 1 & 0 & 0 & \dots  
			\\
			0 & R(1,0) & R(1,1) & 0 & \dots 
			\\
			0 & R(2,0) & R(2,1) & R(2,2) & \dots 
			\\
			\vdots & \vdots & \vdots & \vdots & 
			
			\end{bmatrix}
		}
		\scalemath{0.63}{
			\begin{aligned}
			&\left.\begin{matrix}
			\\
			\\
			\\
			\\[1.5em]
			\\
			\end{matrix}\right\} \text{Out-of-control}%
			\\
			&\left.\begin{matrix}
			\\
			\\
			\\
			\\[1.5em]
			\\
			\end{matrix}\right\} \text{True alarm}%
			\\
			\end{aligned} \nonumber
		}
		\end{matrix}
		\end{equation}
	\end{center}
	It can be seen, that when the process is
	out-of-control without alarm, the distance is not changed. The probabilities for
	the alarm states are calculated using the $R()$ repair size distribution.
	
	So far only the expected cost was considered during the optimisation. In certain
	fields of application the reduction of the cost standard deviation can be
	just as or even more important than the minimisation of the expected cost.
	Motivated by this, let us consider now the weighted average of the cost
	expectation and the cost standard deviation:
	\begin{equation}
	G = pE(C) + (1-p)\sigma(C) \nonumber
	\end{equation}
	Now $G$ is the value to be minimised and $p$ is the
	weight of the expected cost ($0\le p\le 1$). The cost
	standard deviation can easily be calculated by modifying the cost function
	formula. All of the previous models can be used without any significant change,
	one simply changes the value to be minimised from $E(C)$ to $G$.
	
	\subsection{Comparison of Different Scenarios}
	Implementation of the methods was done using the \textbf{\textsf{R}} programming
	language.
	Supplying all the necessary parameters, one can calculate the $G$ value of
	the process for one time unit. It is
	also possible to minimise the $G$ value by finding the optimal time between
	samplings and control limit. All the other parameters
	are assumed to be known. The optimization step can be carried out
	using different tools, the results presented here were obtained with the
	built-in \texttt{optim()} \textbf{\textsf{R}} function: box-constrained
	optimization using PORT routines\cite{Gay}, namely
	the Broyden-Fletcher-Goldfarb-Shanno (L-BFGS-B) algorithm. The optimisation
	procedure can be divided into three steps.
	First, the transition matrix needs to be constructed from the given parameters.
	After this, the stationary distribution of the Markov chain is computed.
	In the third step, the $G$ value is calculated using the stationary
	distribution and the cost function. The optimisation algorithm then
	checks the resulting $G$ value and iterates the previous steps
	with different time between sampling and/or control limit parameters until it
	finds the optimal ones.
	\paragraph{Dependence on Parameters} The testing was done
	using a moderate overall compliance level. This is required because low 
	compliance can result in extreme behaviour in the optimisation such as taking
	the maximum or minimum allowed parameter value. An example for this is when the
	sampling probability depends on both the time between samplings and the state
	of the process: if the compliance level is also relatively low, then given
	certain parameter setups the optimal time between samplings will tend to zero.
	This will essentially create a self-reporting system as the increased distance
	from the target value will eventually increase the compliance and the sampling
	will take place. In a healthcare environment this would mean that the patient
	is told to come back for a control visit as soon as possible, but the patient
	will show up only when the symptoms are severe enough.
	This kind of process behaviour is undesirable in many cases, for example when
	the time between samplings cannot be set arbitrarily.
	
	The results obtained are shown
	in Figure 3.
	\begin{figure}[ht]
		\caption{Optimal parameters and the resulting expected cost and cost standard
			deviation as function of the process standard deviation ($\sigma$),
			out-of-control cost ($c_o$), expected shift size ($\delta$) and weight
			parameter ($p$) for $s=0.2$, $\alpha=1$, $\beta=3$, a$=0.01$, b$=1$,
			$c_s=1$, $c_{rb}=10$, $c_{rs}=10$ \\
			Top: $\sigma=0.1$, Bottom: $\sigma=1$}
		\includegraphics[scale=0.75]{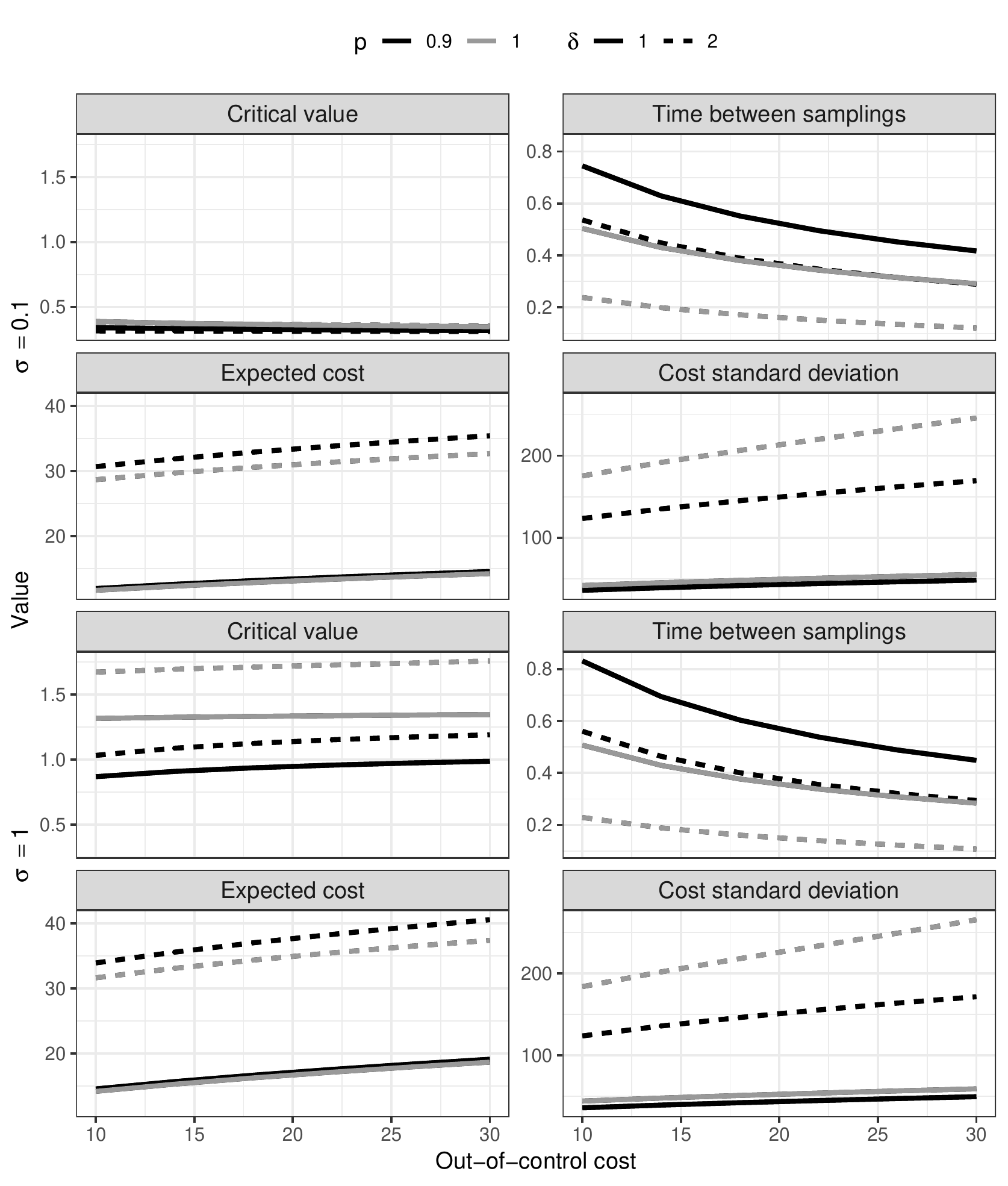}
		\centering
	\end{figure}
	One may observe the weak dependence of the critical value on the
	out-of-control cost. The time between samplings should be decreased with the
	increase of the out-of-control cost, and the average cost and the cost standard
	deviation increase with the out-of-control cost, as expected.
	The effect of the expected shift size on the critical value depends on the
	process standard deviation, as increased shift size results in markedly
	increased critical values only in the case of $\sigma=1$. Higher expected shift
	size entails less time between samplings, and increased expected cost and cost
	standard deviation.
	One can assess, that if the cost standard deviation is taken
	into account during the optimisation procedure ($p=0.9$), then
	lower critical value should be used with increased time between samplings. This
	is logical, because the increased time between sampling will lead to less
	frequent interventions, thus a less erratic process. Of course, at the same
	time we do not want to increase the expected cost, so the critical value is
	lowered. The cost standard deviation is decreased, as expected. What is
	interesting to note is that the expected costs have barely increased compared
	to the $p=1$ setup. This is important, because it shows that by changing the
	parameters appropriately, the cost standard deviation can be lowered - sometimes
	substantially - while the expected cost is only mildly increased. Several
	scenarios entail relatively large cost standard deviations, this is party due to
	the Taguchi-type loss funcion used during the calculations. The process standard
	deviation $\sigma$ has noticeable effect on the critical value only:
	lower critical values should be used for lower standard deviations, as expected.
	
	\paragraph{Sensitising Rules} The effect of sensitising rules\cite{Montgomery}
	was investigated using simulation, since the implementation of theoretical calculations would
	have resulted in a hugely inflated transition matrix which poses a serious
	obstacle in both programming a running times.
	
	Optimal parameters were calculated for $p=0.9$, $\sigma=1$, $s=0.2$, $\delta=2$,
	$\alpha=1$, $\beta=3$, a $=0.01$, b $=1$, $c_s=1$, $c_{rb}=10$, $c_{rs}=10$,
	$c_{rs}=20$. The resulting optimal parameters were $h=0.38$ and $k=1.14$. This
	parameter setup entailed an expected cost of $E(C) = 37.75$ and cost standard
	deviation of $\sigma(C) = 150.33$. The probability of alarm states
	together was $\sum_iP_{r_i}=0.201$. Simulations were run for 50000 sampling
	intervals which equals 19000 unit time. Simulations from the first 100 sampling
	intervals were discarded as it was regarded as a burn-in stage.
	
	First, we present the baseline simulation results - the ones without additional
	rules.
	Overall, the simulation achieved an acceptable level of convergence
	to the theoretical stationary distribution. The empirical expected cost was
	$\overline C = 36.51$. The proportion of alarm states was 0.192. The calculation
	of the empirical standard deviation was carried out by taking into account the
	data of only every 30th sampling interval to deal with the autocorrelation of
	neighboring values. The empirical standard deviation using this method was $s^*
	= 199.37$. It is important to note that the empirical
	results can be somewhat inaccurate, depending on the particular simulation
	results and the parameters used. This is due to the large
	variance and slow convergence of the standard deviation.
	Nonetheless, for this particular scenario the theoretical and empirical results
	were quite close to each other, thus we will compare the effect of sensitising
	rules to this baseline simulation.
	
	The first rule we investigated was the one which produces an alarm in case of
	three consecutive points outside the $\frac{2}{3}k$ warning limit but still
	inside the control limit. Running the simulation
	with the extra rule resulted in $\overline C = 37.42$, $s^* = 171.57$ and a
	ratio of all alarm states together of 0.194, all of these values are within the
	hypothesised confidence interval.
	We can see no major difference in any of these values compared to the baseline.
	
	The second rule was the same as the first one, except this time two consecutive
	points outside the $\frac{2}{3}k$ warning limit were enough to produce an alarm
	signal. The results were $\overline C = 36.54$, $s^* = 190.91$ and 0.200 for the
	proportion of alarm states. Again, no apparent differences can be seen, but it
	is worth noting the proportion of alarm states is somewhat higher in this case
	than the at the baseline or the previous rule, and this was also seen with
	repeated simulation.
	
	Overall, the effect of the investigated sensitising rules seems to be minimal on
	the results. Further investigation is required of the the possible effects in
	case of other parameter setups and rules.

	\section{Application}
	In the following paragraphs we show a real life healthcare example as the
	application of the previously introduced methods. Two approaches will be
	presented: one with and one without taking the standard deviation into account.
	The example problem is to minimise the healthcare burden generated by
	patients with high cardiovascular (CV) event risk. The model is built upon the
	relationship between the low-density lipoprotein (LDL) level and the risk of CV
	events, thus the LDL level is the process of interest.\cite{Boekholdt}
	
	Parameters were estimated from several sources. The list below gives
	information about the meaning of the parameter values in the healthcare
	setting and shows the source of the parameter estimations.
	\begin{center}
		\begin{longtable}{p{3cm}C{3.5cm}p{7cm}}
			\multicolumn{3}{c}{\textbf{Parameter values and sources}}
			\tabularnewline[2ex]
			\toprule
			Parameter value & Meaning & Parameter source \tabularnewline 
			\midrule
			$\mu_0$=3 mmol/l & Target value. & Set according to
			the European guideline for patients at risk.\cite{Garmendia}
			\tabularnewline[1ex]
			\hline
			$\sigma$=0.1 mmol/l & Process standard deviation. &
			Estimated using real life data from
			Hungary, namely registry data through the Healthware Consulting Ltd.
			\tabularnewline[1ex]
			\hline
			$\delta$=0.8/3 & Expected shift size, 0.8 increase
			in LDL per year on average. & Estimated with the help of a
			health professional. \tabularnewline[1ex]
			\hline
			$s$=1/120 & Expected number of shifts in a day, 3 shifts per year on
			average. & Estimated with the help of a
			health professional. \tabularnewline[1ex]
			\hline
			$\alpha = 0.027$, $\beta = 1.15$ & Parameters of the repair size beta
			distribution. & Estimated using an international study which included
			Hungary.\cite{Garmendia} \tabularnewline[1ex]
			\hline
			$q = 0.1$, $z = 30$, & Parameters of the sampling
			probability logistic function. & Patient non-compliance in LDL controlling
			medicine is quite high, and this is represented through the parametrisation
			of the logistic function.\cite{Lardizabal} \tabularnewline[1ex]
			\hline
			$c_s$=\euro5.78 & Sampling cost. & Estimated using the LDL testing cost and
			visit cost in Hungary. \tabularnewline[1ex]
			\hline
			$c_o$=\euro5.30 & Shift-proportional
			daily out-of-control cost. & Estimated using real world data of
			cardiovascular event costs from Hungary
			\tabularnewline[1ex]
			\hline
			$c_{rb}$=\euro11.50 & Base repair cost. & Estimated using the simvastatin
			therapy costs in Hungary \tabularnewline[1ex]
			\hline
			$c_{rs}$=\euro8.63 & Shift-proportional repair cost. & Estimated using the
			simvastatin therapy costs in Hungary \tabularnewline[1ex]
			\bottomrule
		\end{longtable}
	\end{center}
	It is very difficult to give a good estimate for the type and the parameters of
	a distribution that properly models the non-compliance, thus the results here
	can at best be regarded as close approximations to a real life situation. This
	is not a limiting factor, as patients themselves can have vast differences in
	their behaviour, so evaluation of different scenarios are often required, and
	will also be presented here. Since high LDL levels rarely produce noticeable
	symptoms, the sampling probability only depends on the time between samplings,
	thus the sampling probability was modeled by the logistic function and not by
	the beta distribution.\cite{cholesterol} It is important to note that the
	proportional costs increase according to a Taguchi-type loss function, thus huge
	expenses can be generated if the patient's health is highly out-of-control.
	
	\paragraph{Optimisation using only the cost expectation}
	The optimal parameters for the case when the cost standard deviation was not
	taken into account were $56.57$ days and $0.143$ mmol/l for the time between
	samplings and the critical increase in the LDL level from the guideline value
	respectively.
	These parameters entailed an average daily cost of \euro$0.469$ and standard
	deviation of \euro$0.562$.
	This result is interesting, because the optimisation says that we should use a
	somewhat higher critical value than the one according to the guideline - 0
	mmol/l critical increase would be the original 3 mmol/l value - but we should
	monitor the patients more often than usual - times of the LDL measurements are
	usually several months or years apart.
	It is important to note, that this is a strictly cost effective viewpoint which
	could be overwritten by a health professional. Nonetheless, the results provide
	a possible new approach to the therapeutic regime for controlling LDL level.
	Often, it is good to look at the interaction between the parameters and the
	resulting average cost, especially in situations where the optimal parameters
	cannot be used because of real life reasons.
	\begin{figure}[ht]
		\caption{Expected cost as function of  the time between samplings and the
			critical value}
		\includegraphics[scale=0.75]{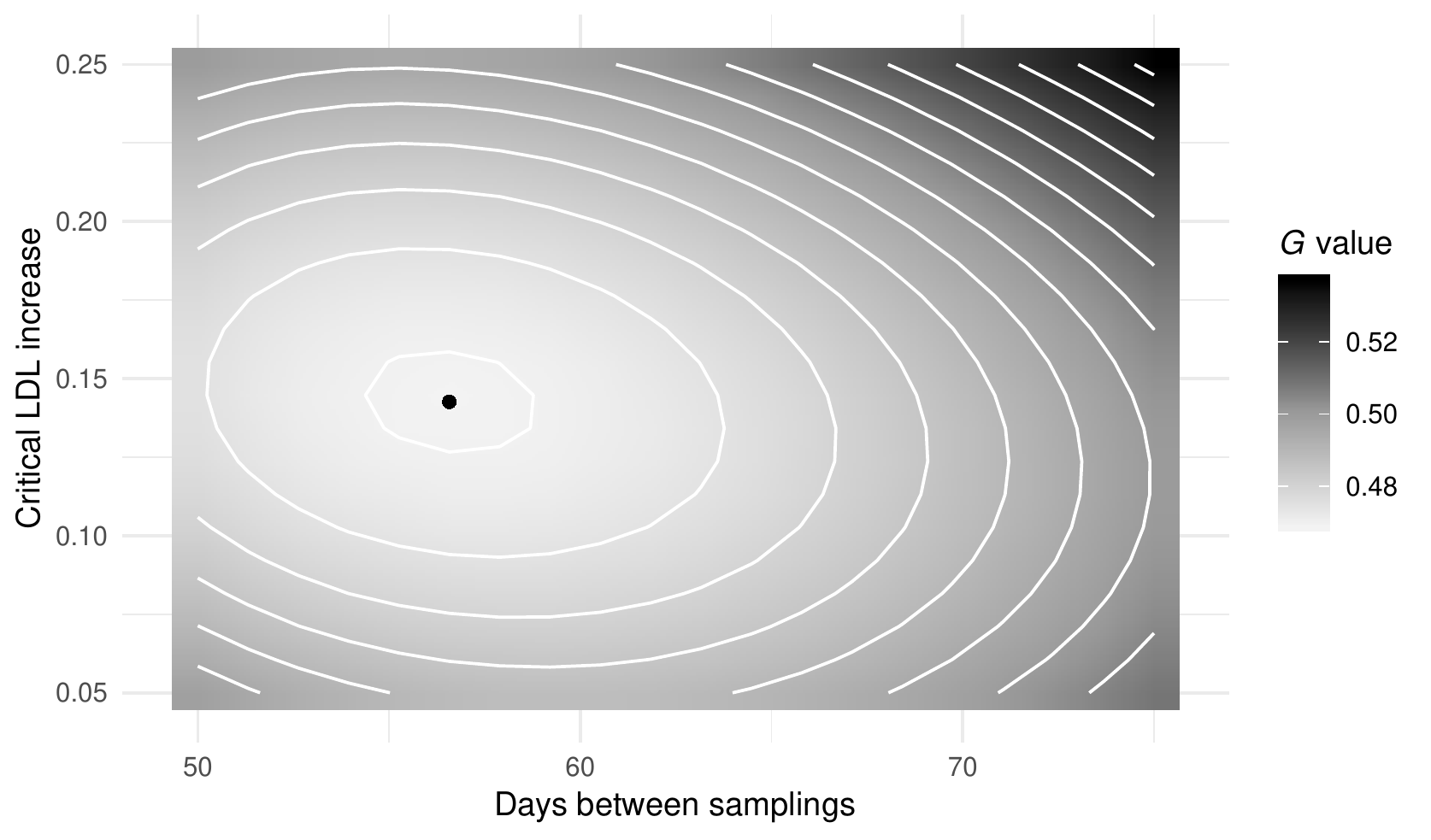}
		\centering
	\end{figure}
	The heat map of Figure 4 shows the
	average cost as the function of the different parameter values.
	The dot in the lightest area of the figure corresponds to the optimal cost. Any
	other point entails a higher average cost.
	It can clearly be seen that too low or high critical values will both increase
	the average daily cost. What is interesting - for this scenario -
	is that the change in the time between samplings entails relatively low change
	in the critical LDL increase: even if the time between control visits is changed
	the critical value should stay around $0.12 - 0.18$ mmol/l.
	
	\paragraph{Optimisation using cost expectation and cost standard deviation}
	In this part, the cost standard deviation is also taken into account with
	$p=0.9$, thus the weight of the standard deviation in the calculation of $G$ is
	0.1.
	The optimal parameters found by our approach were $64.76$ days and $0.129$
	mmol/l for the time between samplings and critical increase in the LDL level
	respectively. These parameters entailed an average daily cost of \euro$0.477$
	and standard deviation of \euro$0.418$. The inclusion of the cost standard
	deviation into the model has somewhat increased the time between control visits
	and decreased the critical value. The expected cost somewhat increased,
	while the cost standard deviation was moderately decreased. Figure 5 shows the
	previous heat map with non-compliance included in the model.
	\begin{figure}[ht]
		\caption{$G$ value as function of  the time between samplings and the
			critical value}
		\includegraphics[scale=0.75]{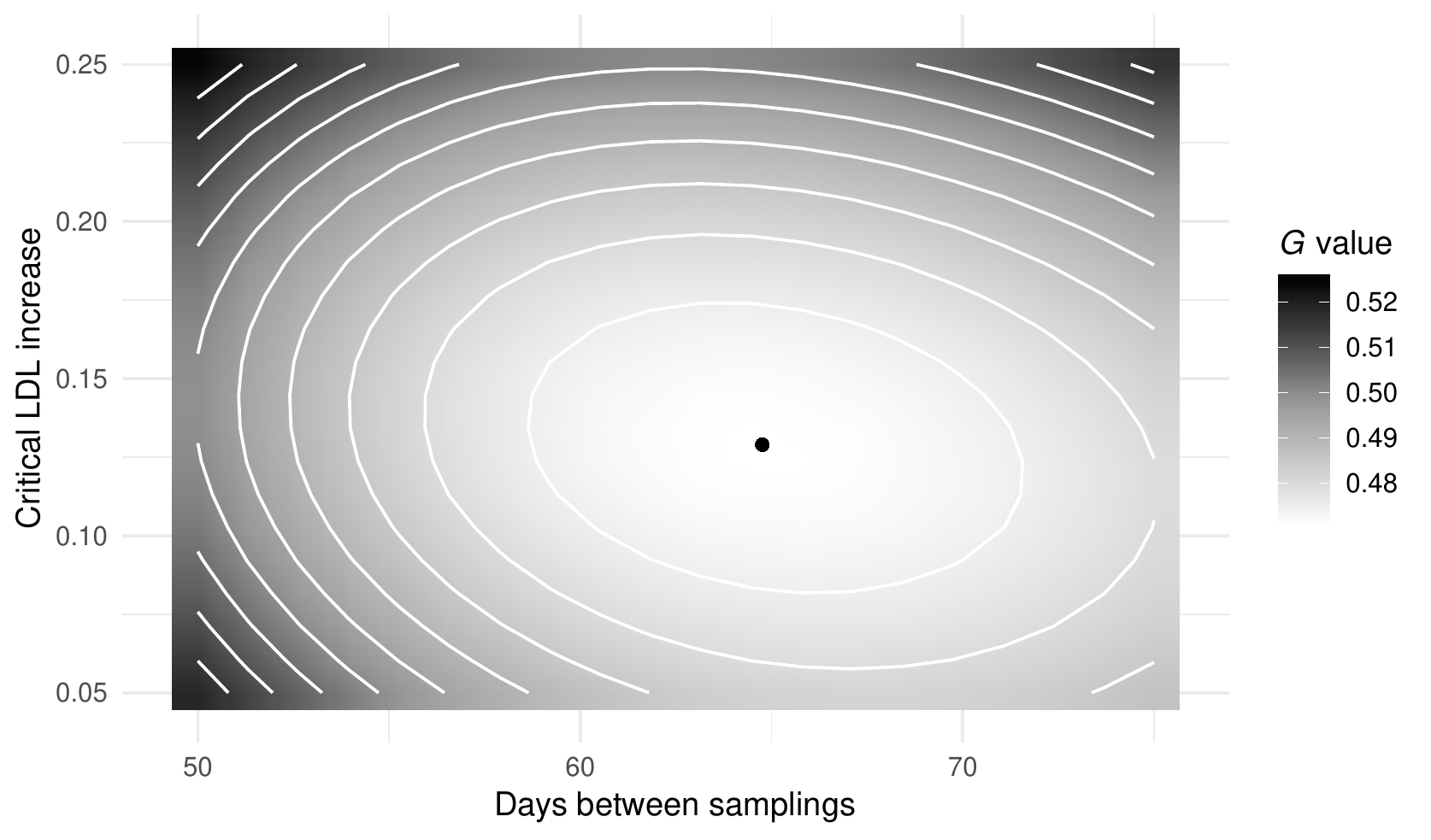}
		\centering
	\end{figure}
	It can be seen that the elliptical shape of the heat map has not changed: the
	change in the time between control visits still does not entail great change in
	the critical value.

	\paragraph{Sensitivity Analysis} As there were uncertainty about the estimation
	of several parameters, it is important to assess the effect of different
	parameter setups. The results for different out-of-control costs are
	plotted for both approaches.
	The results can be seen in Figure 6.
	\begin{figure}[H]
		\caption{Parameters, average total cost and cost standard deviation as function
			of the out-of-control cost}
		\includegraphics[scale=0.75]{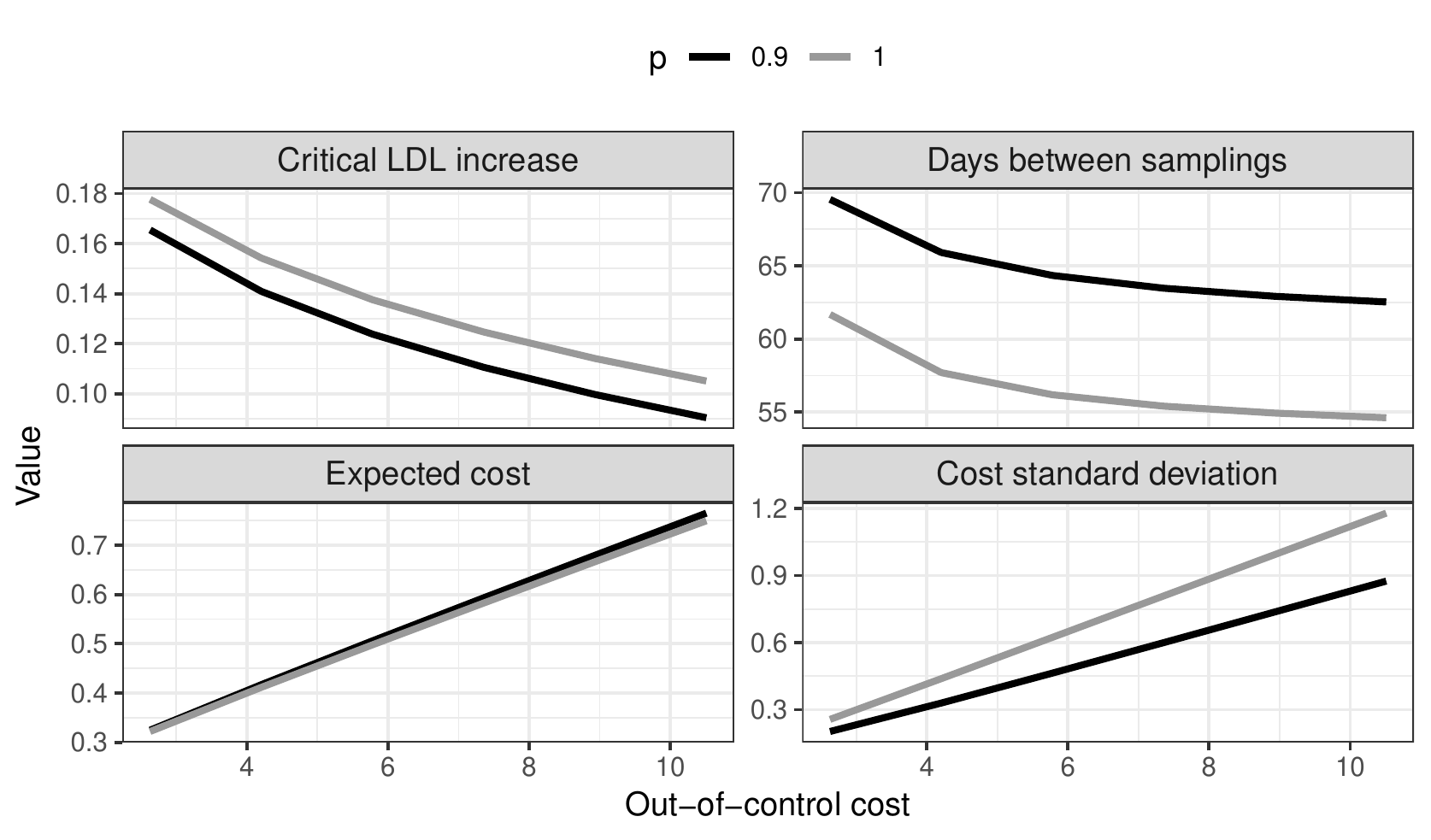}
		\centering
	\end{figure}
	The critical value and time between
	samplings decrease and the average cost and cost standard deviation increase
	with higher out-of-control costs. Just as on the heat maps, one can
	observe here, that if the cost standard deviation is taken into account in the
	optimisation, then the critical value should be lowered and the time between
	samplings increased. One can observe, that a substantial decrease can be
	achieved in the cost standard deviation while the cost expectation barely
	changes.
	
	Uncertainty was also high around the estimation of the sampling probability. The
	sigmoid's midpoint so far was $z=30$ days, meaning that the probability of
	sampling was $0.5$ at $h=30$ and increased with $h$. Figure 7 contains results
	for different $z$ values.
	\begin{figure}[ht]
		\caption{Top: parameters, average total cost and cost standard deviation as
			function of the sigmoid's midpoint ($z$), Bottom left: distance distributions:
			lighter lines corresponds to greater $z$ values, Bottom right: sampling
			probability as function of $z$, for $h=64.76$ \\}
		\includegraphics[scale=0.75]{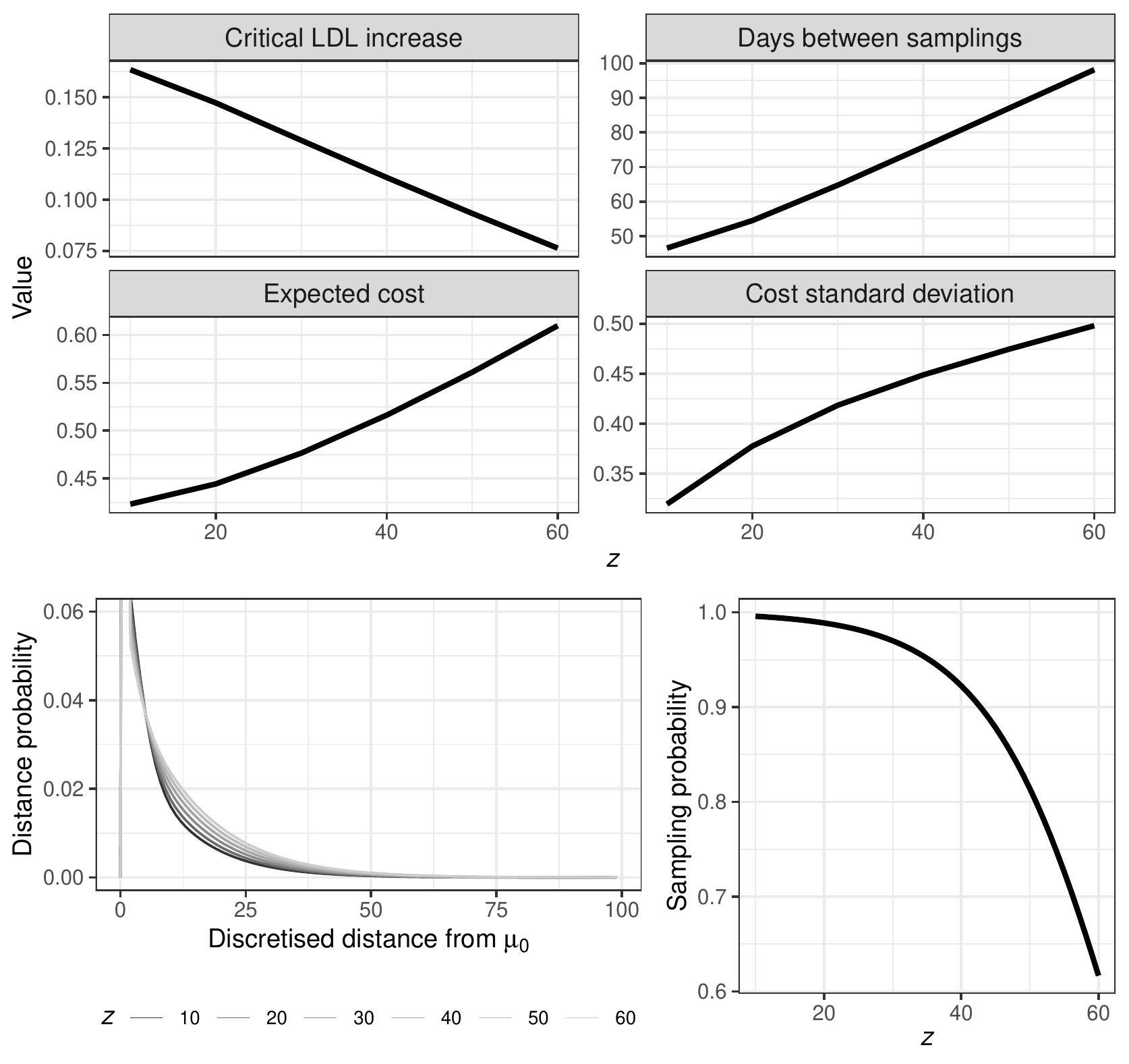}
		\centering
	\end{figure}
	One can observe, that as the probability of successful sampling decreases - the
	value of $z$ is increased - the critical value decreases and the time between
	samplings increases. This can be explained by the increased uncertainty of the
	sampling: More time between samplings entails higher patient compliance,
	and when the visit occurs a low critical value is used to ensure treatment.
	The cost expectation and standard deviation increases with lower sampling
	probabilities. There are only minor differences between the stationary
	distributions, nonetheless it can be seen that lower sampling probability is
	associated with higher probabilities for greater distances from the target
	value. The last panel shows how the sampling probability decreases with
	increasing $z$ values for a fixed $h$.
	
	\section{Conclusions}
	
	Cost-optimal control charts based on predominantly Duncan's cycle model are
	found in a wide variety of areas. Even though the benefit of using these charts in
	industrial and engineering settings is unquestionable, the numerous assumptions
	needed about the processes makes the applicability of the traditional models
	problematic in certain environments. Motivated by the desire to apply
	cost-optimal control charts on processes with erratic behaviour and imperfect
	repair mechanism - such as ones found in healthcare -  this paper presented a
	Markov chain-based framework for the generalisation of these control charts,
	which enabled the loosening of some of the usual assumptions.
	
	Cost-optimisation of control charts are usually carried out by finding the
	optimal critical value, time between samplings and sample size. Our work
	concentrated on the monitoring of a single element at a time - e.g. a patient -
	thus the methods presented here always used a sample size of $1$.
	
	Building on and expanding the work of Zempléni et al.\cite{Zempleni} we
	discussed three types of generalisations: the random shift size, the imperfect
	repair and the non-compliance. The random shift size means that only the
	distribution of the shift size and its parameters are assumed to be known. This
	let us monitor processes which are potentially drifting in nature. The second
	generalisation - the imperfect repair - assumed that the process stays out of
	control even after repair, but on a level closer to the target value than
	before. This type of repair response is often observed in treatments in
	healthcare, but may be found in other areas too. The third generalisation was
	intended to help the modeling of patient or staff non-compliance. We
	implemented this concept in a way that allows sampling times to be skipped by a
	probability governed by a distribution or function with known parameters.
	
	Since the processes modeled with the above loosened assumptions can create
	complicated trajectories between samplings, the mathematical description of
	these was also necessary. We proposed an expectation calculation method of a
	function of the values taken on by the process between samplings. The
	application of this proposition while assuming exponentially distributed shift
	sizes, Poisson distributed event numbers and Taguchi-type loss function yielded
	a compact formula for the expectation.
	
	We implemented our theoretical results in the \textbf{\textsf{R}} programming
	language and investigated the effect of parameter estimation uncertainty on the
	optimal parameters and the resulting expected cost and cost standard deviation.
	We also tested the effect of involving the cost standard deviation in the
	optimisation procedure itself. We found that typically the critical value
	increases and the time between samplings decreases with the expected shift size.
	Also, higher expected shift sizes entail higher expected costs and cost standard
	deviations.
	It was seen that with the increase of the out-of-control cost - in most cases -
	the critical value stagnated, the time between samplings decreased, and the
	expected cost and the cost standard deviation increased. The involvement of
	the cost standard deviation in the optimisation procedure lowered the standard
	deviation while the cost expectation barely increased. We have found no evidence
	that sensitising rules - such as values outside the warning limit - would
	change the results substantially.
	
	We presented an example of real-life application involving low-density
	lipoprotein monitoring. The results indicated that the cost-optimal critical
	value is somewhat higher and the cost-optimal time between control visits is
	less than the ones usually used according to medical guidelines.
	
	In the era of Industry 4.0, the cost-optimal control charts presented here can
	be applied to a wider range of processes than the traditional ones.
	Nonetheless there are still areas worth investigating. One of the features
	still missing is the proper modeling of the repair procedure, since it was
	assumed to be an instantaneous event, which may not be appropriate in many
	situations. The mathematical elaboration of a continuous model involving e.g.
	time series could also be beneficial.
	
	\section{Acknowledgements}
	The authors would like to express their gratitude towards the Healthware
	Consulting Ltd. (Budapest, Hungary) for their help in providing data and
	professional knowledge during the writing of the Application section.
	
	The work was supported by the project EFOP-3.6.2-16-2017-00015, which was
	supported by the European Union and co-financed by the European Social Fund.
	

	
\end{document}